\documentclass[leqno, 12pt]{article}
\usepackage{times}
\usepackage{amssymb,amscd}
\usepackage{amsmath}
\usepackage{amsfonts}
\usepackage{mathrsfs}

\usepackage{amsthm}

\newtheorem{theorem}{Theorem}[section]
\newtheorem{definition}[theorem]{Definition}
\newtheorem{lemma}[theorem]{Lemma}
\newtheorem{corollary}[theorem]{Corollary}

\begin{document}

\date{}

\title{Epistemic extensions of combined classical and intuitionistic propositional logic}
\author{Steffen Lewitzka
\thanks{Instituto de Matem\'atica,
Departamento de Ci\^encia da Computa\c c\~ao,
Universidade Federal da Bahia UFBA,
40170-110 Salvador -- BA,
Brazil,
e-mail: steffenlewitzka@web.de}}
\maketitle

\begin{abstract}
Logic $L$ was introduced by Lewitzka \cite{lewjlc2} as a modal system that combines intuitionistic and classical logic: $L$ is a conservative extension of CPC and it contains a copy of IPC via the embedding $\varphi\mapsto\square\varphi$. In this article, we consider $L3$, i.e. $L$ augmented with S3 modal axioms, define basic epistemic extensions and prove completeness w.r.t. algebraic semantics. The resulting logics combine classical knowledge and belief with intuitionistic truth. Some epistemic laws of \textit{Intuitionistic Epistemic Logic} studied by Artemov and Protopopescu \cite{artpro} are reflected by classical modal principles. In particular, the implications ``\textit{intuitionistic truth} $\Rightarrow$ \textit{knowledge} $\Rightarrow$ \textit{classical truth}" are represented by the theorems $\square\varphi\rightarrow K\varphi$ and $K\varphi\rightarrow\varphi$ of our logic $EL3$, where we are dealing with classical instead of intuitionistic knowledge. Finally, we show that a modification of our semantics yields algebraic models for the systems of Intuitionistic Epistemic Logic introduced in \cite{artpro}.
\end{abstract}

\section{Introduction}

The approach presented in this paper relies on the assumptions of R. Suszko's non-Fregean logic (see, e.g., \cite{blosus, sus1}) and on results of our recent research \cite{lewjlc2, lewjlc1, lewsl}. It is inspired by \cite{lewsl1} and by ideas coming from Intuitionistic Epistemic Logic \cite{artpro}. A non-Fregean logic contains an identity connective and formulas of the form $\varphi\equiv\psi$ expressing that $\varphi$ and $\psi$ have the same meaning, denotation. The basic classical non-Fregean logic is the Sentential Calculus with Identity SCI \cite{blosus} which can be axiomatized by classical propositional logic CPC along with the following identity axioms:\\
\noindent (Id1) $\varphi\equiv\varphi$\\
(Id2) $(\varphi\equiv\psi)\rightarrow (\varphi\leftrightarrow\psi)$\\
(Id3) $(\varphi\equiv\psi)\rightarrow (\chi[x:=\varphi]\equiv\chi[x:=\psi])$\\
where $\varphi[x:=\psi]$ is the formula that results from substituting $\psi$ for every occurrence of variable $x$ in $\varphi$. We refer to (Id1)--(Id3) as the \textit{axioms of propositional identity}.\footnote{Instead of (Id3), Suszko presents a set of three axioms which together are equivalent to (Id3). By a \textit{proposition} we mean the denotation of a formula. Instead of \textit{proposition}, Suszko uses the term \textit{situation}.} According to G. Frege, the meaning of a formula is given by its truth value. This can be formalized by a scheme called by Suszko the \textit{Fregean Axiom}: $(\varphi\leftrightarrow\psi)\rightarrow (\varphi\equiv\psi)$, i.e. two formulas denote the same proposition whenever they have the same truth value. In non-Fregean logics, the Fregean Axiom is not valid.

Recall that C.I. Lewis' modal systems S1--S3 were originally designed as axiomatizations of laws for strict implication $\square(\varphi\rightarrow\psi)$ (see, e.g., \cite{hugcre} for a discussion). One immediately recognizes that all Lewis systems S1--S5 satisfy the axioms (Id1) and (Id2) of propositional identity if $\varphi\equiv\psi$ is defined as strict equivalence $\square(\varphi\rightarrow\psi)\wedge\square(\psi\rightarrow\varphi)$.\footnote{Certain connections between non-Fregean logic and the modal systems S4 and S5 were already investigated by Suszko and Bloom (see \cite{blosus, sus}).} In \cite{lewjlc1, lewsl} we proved that S3 is the weakest Lewis modal logic where strict equivalence satisfies all axioms of propositional identity (Id1)--(Id3). Moreover, in \cite{lewjlc1} we showed that logic S1+SP, i.e. S1 augmented with (Id3) as theorem scheme, which we also call the Substitution Principle SP, has a simple algebraic semantics.\footnote{Recall that there is no known intuitive semantics for S1.} With system S1+SP we proposed a formalization of strict equivalence as propositional identity in the sense of Suszko's intuitive axioms (Id1)--(Id3) above. We were able to show that S1+SP is distinct from S2 and is strictly contained in S3.

S1+SP can be regarded as the weakest modal system that combines Lewis' approach to strict implication with the principles of Suszko's non-Fregean logic SCI. Proceeding from these assumptions, we proposed in \cite{lewjlc2} the modal logic $L$, which is axiomatized by intuitionistic logic IPC, the modal axioms of S1 together with a modal axiom expressing the constructive character of instuitionistic truth, and SP and \textit{tertium non datur} as theorems. $L$ is a \textit{classical} logic in which the modal operator $\square$ plays the role of an intuitionistic truth predicate. Intuitionistically equivalent formulas are identified and denote the same proposition (e.g. $\neg\varphi\equiv\neg\neg\neg\varphi$ is a theorem), while classically equivalent formulas have in general different meanings ($\varphi\equiv\neg\neg\varphi$ is not a theorem, but $\varphi\leftrightarrow\neg\neg\varphi$ is a theorem). $L$ is a conservative modal extension of CPC and contains a copy of IPC via the embedding $\varphi\mapsto\square\varphi$. In this sense, $L$ combines CPC and IPC. In the present paper, we consider $L3$, the $S3$-version of $L$, and investigate basic epistemic extensions. In this setting of combined classical and intuitionistic logic, the question arises in which way classical epistemic laws of knowledge and belief relate to the modal laws of $L3$. We are inspired by the approach developed in \cite{lewsl1} where a non-Fregean semantics for the modeling of epistemic properties is presented. Knowledge or belief of an agent is modeled in a natural way as a set of propositions, i.e. as a subset $\mathit{BEL}\subseteq M$ of the universe $M$ of an algebraic (non-Fregean) model. In the cases of our logics $L$ and $L3$, a non-Fregean model is a Heyting algebra where classical truth is represented by an ultrafilter $\mathit{TRUE}$, and intuitionistic truth is given by the top element of the underlying lattice. Recall that all intuitionistic tautologies denote the top element of any given Heyting algebra, under all assignments. We assume that also our modal axioms are intuitionistically acceptable and therefore should denote the top element, too. Furthermore, we also assume that at least all axioms should be known by the agent. Hence, the top element is a known proposition, i.e. an element of $\mathit{BEL}$. This can be expressed by $\square\varphi\rightarrow K\varphi$. The semantical condition $\mathit{BEL}\subseteq\mathit{TRUE}$ corresponds to the facticity of knowledge, i.e. to the theorem $K\varphi\rightarrow\varphi$. Many further properties of knowledge and belief, such as the distribution law $K(\varphi\rightarrow\psi)\rightarrow (K\varphi\rightarrow K\psi)$, can be modeled semantically by imposing suitable closure conditions on the set $\mathit{BEL}$ of each model.\footnote{In \cite{lewsl1} are modeled also more complex epistemic concepts such as common knowledge in a group of agents.}

Special attention deserves the bridge theorem $\square\varphi\rightarrow K\varphi$ which establishes the connection between the modal and the epistemic part of the logic. Actually, we will need a slightly stronger bridge axiom, namely $\square\varphi\rightarrow\square K\varphi$, in order to warrant that the axioms of propositional identity, in particular (Id3), also hold in the extended epistemic language (see Lemma \ref{100}). That bridge axiom, together with the rule of Axiom Necessitation, yields $\square(\square\varphi\rightarrow\square K\varphi)$, which can be regarded as a representation of the Brouwer-Heyting-Kolmogorov (BHK) reading of \textit{intuitionistic co-reflection} $\varphi\rightarrow K\varphi$, an axiom of the systems $IEL^-$ and $IEL$ of Intuitionistic Epistemic Logic introduced by Artemov and Protopopescu \cite{artpro}. $IEL^-$ and $IEL$ are intuitionistic logics where truth is understood as proof, and epistemic laws are in accordance with the constructive BHK semantics of intuitionistic logic. The intuitionistically unacceptable principle of \textit{reflection} $K\varphi\rightarrow\varphi$ is replaced with \textit{intuitionistic reflection} $K\varphi\rightarrow\neg\neg\varphi$, i.e. ``known propositions cannot be false". We adopt the latter as an axiom and derive classical reflection. Then the implications ``intuitionistic truth $\Rightarrow$ (intuitionistic) knowledge $\Rightarrow$ classical truth", underlying the approach of \cite{artpro}, can be represented by theorems $\square\varphi\rightarrow K\varphi$ and $K\varphi\rightarrow\varphi$ of our modal logics, where, of course, we are dealing with \textit{classical} instead of \textit{intuitionistic} knowledge. By introducing additional modal axioms, we are able to establish further laws for the reasoning about classical knowledge and intuitionistic truth. 

\section{Deductive systems}

The language is inductively defined in the usual way over an infinite set of variables $x_0, x_1, ... $, logical connectives $\wedge$, $\vee$, $\rightarrow$, $\bot$, the modal operator $\square$ and the epistemic operator $K$. $Fm$ denotes the set of all formulas and $Fm_0\subseteq Fm$ is the set of \textit{propositional formulas}, i.e. formulas that neither contain the modal operator $\square$ nor the epistemic operator $K$. We shall use the following abbreviations:\\

\noindent $\neg\varphi :=\varphi\rightarrow\bot$\\
$\top :=\neg\bot$\\
$\varphi\leftrightarrow\psi :=(\varphi\rightarrow\psi)\wedge (\psi\rightarrow\varphi)$\\
$\varphi\equiv\psi := \square(\varphi\rightarrow\psi)\wedge\square(\psi\rightarrow\varphi)$ (``propositional identity = strict equivalence") \\
$\square\Phi :=\{\square\psi\mid\psi\in\Phi\}$, for $\Phi\subseteq Fm$\\

We consider the following \textbf{Axiom Schemes} (INT) and (A1)--(A8)\\

\noindent (INT) all theorems of IPC and their substitution-instances\footnote{A substitution-instance of $\varphi$ is the result of uniformly replacing variables in $\varphi$ by formulas of $Fm$.}\\
(A1) $\square(\varphi\vee\psi)\rightarrow(\square\varphi\vee\square\psi)$ \\
(A2) $\square\varphi\rightarrow\varphi$\\
(A3) $\square(\varphi\rightarrow\psi)\rightarrow\square(\square\varphi\rightarrow\square\psi)$\\
(A4) $\square\varphi\rightarrow\square\square\varphi$\\
(A5) $\neg\square\varphi\rightarrow\square\neg\square\varphi$\\
(A6) $K(\varphi\rightarrow\psi)\rightarrow (K\varphi\rightarrow K\psi)$ (distribution of belief)\\
(A7) $\square\varphi\rightarrow \square K\varphi$ (co-reflection)\\
(A8) $K\varphi\rightarrow\neg\neg\varphi$ (intuitionistic reflection)\\

\noindent and the following \textbf{Theorem Scheme} (T) of \textit{tertium non datur}\\

\noindent (T) $\varphi\vee\neg\varphi$\\

The inference rules are Modus Ponens MP ``From $\varphi$ and $\varphi\rightarrow\psi$ infer $\psi$", and Axiom Necessitation AN ``If $\varphi$ is an axiom, then infer $\square\varphi$". \\

The intended meaning of a formula $\square\varphi$ is: ``there is a proof of $\varphi$" (i.e. $\varphi$ is proved), where proof is understood as intuitionistic truth in the sense of the BHK semantics of intuitionistic logic. Accordingly, $\neg\square\varphi$ means ``$\varphi$ is not proved" (although a proof might be possible), and $\square\neg\varphi$ reads as ``a proof of $\varphi$ is impossible". The above modal axioms then can be read as principles for the reasoning about intuitionistic truth and (classical) knowledge. The basic modal axioms are (A1)--(A3) together with the usual distribution axiom for knowledge (A6) and the bridge axiom (A7). (A2) and (A3) are axioms of Lewis' system S3 of strict implication, given in the style of Lemmon (see, e.g., \cite{hugcre}). Axiom (A1) says that to prove $\varphi\vee\psi$ it is necessary to prove $\varphi$ or to prove $\psi$, which is in line with the BHK semantics.\footnote{($\square\varphi\vee\square\psi)\rightarrow\square( \varphi\vee\psi)$ is derivable in our systems.} Of course, we also expect that proved propositions are classically true. This is expressed by (A2). Recall that, according to the BHK intepretation, a proof of an implicative formula $\varphi\rightarrow\psi$ is given by a (not further specified) \textit{construction} that transforms any given proof of $\varphi$ into a proof of $\psi$. So if we state $\square(\varphi\rightarrow\psi)$, then we assume the existence of such a construction. Of course, if a construction with those properties exists, then we have in particular evidence (i.e. a proof) of the fact that whenever $\varphi$ is proved, there is also a proof of $\psi$. This is the situation described by (A3).  
In this sense, (A3) can be seen as an attempt to translate the constructive content of the BHK interpretation of $\varphi\rightarrow\psi$ into a classical reading. The axioms (A4) and (A5) describe further laws for the reasoning about proofs and correspond to axioms of Lewis' systems S4 and S5, respectively. Finally, (A7) and (A8) are related to axioms of Intuitionistic Epistemic Logic \cite{artpro} and will be discussed below.\\

The distinction between \textit{axioms} and \textit{theorems} in the definition of our deductive systems is important. While all axioms are regarded as intuitionistically acceptable principles, the classical principle of \textit{tertium non datur} is introduced as a theorem scheme. Note that the inference rule AN applies only to axioms.\\ 

If instead of the axioms (INT) we consider (CL) (i.e., all theorems of CPC and their substitution-instances), (A2) and the following (A3')
\begin{equation*}
\square(\varphi\rightarrow\psi)\rightarrow(\square(\psi\rightarrow\chi)\rightarrow\square(\varphi\rightarrow\chi)),
\end{equation*}
as axiom schemes, along with the rules of MP and the Substitution of Proved Strict Equivalents SPSE ``If $(\varphi\equiv\psi)$ is a theorem, then $(\chi[x:=\varphi]\equiv\chi[x:=\psi])$ is a theorem", then we obtain Lewis' modal system S1 (see, e.g., \cite{hugcre} for a similar presentation). Stronger than rule SPSE is the following axiom scheme (Id3) of propositional identity, which we also call the Substitution Property SP:
\begin{equation*}
(\varphi\equiv\psi)\rightarrow (\chi[x:=\varphi]\equiv\chi[x:=\psi]).
\end{equation*}
The system S1+SP, introduced and studied in \cite{lewjlc1}, results from S1 by adding all formulas of the form SP as theorems. Recall that (A3) is the essential axiom scheme of S3, i.e. S3 results from S2 (or even from S1) by adding (A3) as an axiom. By $\square$SP we refer to the collection of all instances of scheme SP prefixed by operator $\square$. System S1+$\square$SP then results from S1 by adding all instances of $\square$SP as theorems. We saw in \cite{lewjlc1} that S1+SP $\subsetneq$ S1+ $\square$SP $\subseteq$ S3. In particular, all instances of $\square$SP are derivable in S3, and S3 is the weakest among Lewis' modal logics with that property. The question whether S1+$\square$SP equals S3, however, is left open -- we believe it can be answered positively. Modal system $L$, introduced in \cite{lewjlc2}, results from S1+SP by replacing (CL) with (INT), adding (A1) as an axiom scheme and adding (T) as a theorem scheme. $L$ is the weakest modal logic that contains the basic laws of Lewis' strict implication and Suszko's non-Fregean logic, and combines IPC and CPC in the sense of \cite{lewjlc2}. Note that $L$ contains the instances of SP as theorems instead of axioms. However, since $\square$ is intended as a predicate for intuitionistic truth, it seems to be reasonable to extend logic $L$ in the sense that rule AN applies to \textit{all} intuitionistically acceptable formulas, including the instances of SP. We know that $\square$SP is contained in S3. Moreover, it is technically easier to work with S3-axioms instead of scheme $\square$SP. For these reasons, we consider in this paper the logic $L3$ which results from $L$ by adding the S3-axiom (A3). That is, $L3$ is given by the axiom schemes (INT) and (A1)--(A3), the theorem scheme (T) and the inference rules of MP and AN. We define the following epistemic extensions of $L3$:

\begin{itemize}
\item $EL3^- = L3 + (A7)$ 
\item $EL3=EL3^- + (A8)$
\item $EL4=EL3 + (A4)$
\item $EL5=EL4 + (A5)$
\end{itemize}

The notion of derivation is defined as usual. For $\Phi\cup\{\varphi\}\subseteq Fm$, we write $\Phi\vdash_\mathcal{L}\varphi$ if there is a derivation of $\varphi$ from $\Phi$ in logic $\mathcal{L}$. \\

Note that by (A2) and (A3), the distribution law of normal modal logis, also called axiom K, $\square(\varphi\rightarrow\psi)\rightarrow (\square\varphi\rightarrow\square\psi)$, is a theorem of all our epistemic logics.
We showed in \cite{lewjlc1, lewsl} that SP (and even $\square$SP) derives in S3. Roughly speaking, SP ensures that propositional identity $\equiv$ defines a congruence relation on the set of modal formulas -- i.e. an equivalence relation that respects the connectives and operators of the underlying modal language. This is a useful, if not necessary, condition for the construction of a natural algebraic semantics. Now we are working with an extended propositional language which besides the modal operator contains an epistemic operator. We have to show that SP still holds in the extended language of our epistemic logics. For the proof we will need S3-axiom (A3) as well as the axioms of co-reflection and distribution of knowledge. 

\begin{lemma}\label{100}
Scheme SP holds: $\vdash_{EL3^-} (\varphi\equiv\psi)\rightarrow (\chi[x:=\varphi]\equiv\chi[x:=\psi])$.
\end{lemma}

\begin{proof}
It is enough to show that the following formulas are theorems:\\
\noindent (a) $((\varphi_1\equiv\psi_1)\wedge (\varphi_2\equiv\psi_2)) \rightarrow (\varphi_1 * \varphi_2)\equiv (\psi_1 * \psi_2)$, for $*\in \{\vee,\wedge,\rightarrow\}$\\
(b) $(\varphi\equiv\psi)\rightarrow (\square\varphi\equiv\square\psi)$\\
(c) $(\varphi\equiv\psi)\rightarrow (K\varphi\equiv K\psi)$\\
The assertion then follows by induction on the complexity of formula $\chi$. We proved in \cite{lewsl, lewjlc2} that in modal system S3 the relation of strict equivalence $\equiv$ satisfies Suszko's axioms of propositional identity (i.e. essentially (a)) and additionally (b). In the following we argue similarly. We will use the fact that $\square(\varphi\wedge\psi)\leftrightarrow (\square\varphi\wedge\square\psi)$ is a theorem of S3 and it derives in the same way in $L3$ and our epistemic logics. Assume $*$ is the connective $\rightarrow$ and note that $((\varphi_1\leftrightarrow\psi_1)\wedge (\varphi_2\leftrightarrow\psi_2)) \rightarrow (\varphi_1 \rightarrow \varphi_2)\leftrightarrow (\psi_1 \rightarrow \psi_2)$ is a theorem of $IPC$. Then rule AN and distribution yield $((\varphi_1\equiv\psi_1)\wedge (\varphi_2\equiv\psi_2)) \rightarrow (\varphi_1 \rightarrow \varphi_2)\equiv (\psi_1 \rightarrow \psi_2)$. The cases of the remaining logical connectives are shown analogously. The formulas of (b) derive with the help of the S3-axiom (A3). Finally, we consider (c) which involves the new operator $K$. We start with $\varphi\equiv\psi$, i.e. $\square(\varphi\rightarrow\psi)\wedge\square(\psi\rightarrow\varphi)$. By co-reflection, this implies $\square K(\varphi\rightarrow\psi)\wedge \square K(\psi\rightarrow\varphi)$ which in turn implies $\square(K\varphi\rightarrow K\psi)\wedge \square(K\psi\rightarrow K\varphi)$ (apply rule AN to axiom (A6) and then apply distribution), i.e. $K\varphi\equiv K\psi$. Thus, $(\varphi\equiv\psi)\rightarrow (K\varphi\equiv K\psi)$ is a theorem of all our epistemic logics.
\end{proof}

The following result expresses the fact that (in every given model) there is exactly one necessary proposition -- the proposition denoted by $\top$. For a proof that relies on SP, we refer the reader to [\cite{lewjlc1}], Lemma 2.3]. In our modal logics, the proposition denoted by $\top$ stands for intuitionistic truth. Then the following scheme of biconditionals says in particular that $\square\varphi$ is classically true iff $\varphi$ holds intuitionistically.

\begin{lemma}\label{110}
$\vdash_{EL3^-}\square\varphi\leftrightarrow (\varphi\equiv\top)$.
\end{lemma}

The systems $IEL^-$ and $IEL$ of Intuitionistic Epistemic Logic studied in \cite{artpro} are axiomatized by (INT), (A6), intuitionistic co-reflection $\varphi\rightarrow K\varphi$ instead of (A7), and, only in case of $IEL$, additionally (A8). The only inference rule is MP.

In $IEL^-$ and $IEL$, the formalized epistemic principles harmonize with the constructive BHK reading. The axiom of intutionistic co-reflection $\varphi\rightarrow K\varphi$ is evident under the assumptions that intuitionistic truth is proof, proof yields (a strict kind of) verification, and verification yields (intuitionistic) knowledge/belief. While the facticity axiom $K\varphi\rightarrow\varphi$ of classical knowledge cannot be justified under the BHK reading, intuitionistic reflection, $K\varphi\rightarrow\neg\neg\varphi$ is acceptable. In this sense, $IEL^-$ and $IEL$ are intuitionistic logics of \textit{intuitionistic} knowledge and belief. \\

The logics presented in this paper are logics of \textit{classical} knowledge and belief: for example, $K\varphi\vee\neg K\varphi$ and $K\varphi\rightarrow\varphi$ are theorems of $EL3$. The systems result from the addition of classical epistemic principles to $L3$ as a classical modal logic for the reasoning about intuitionistic truth. The bridge axiom (A7) 
\begin{equation}\label{classco}
\square\varphi\rightarrow \square K\varphi,
\end{equation}
plays a key role in this setting of combined classical epistemic and intuitionistic logic. Since it is an axiom, we may apply rule AN and obtain
\begin{equation}\label{classco2}
\square(\square\varphi\rightarrow \square K\varphi),
\end{equation}
which, in a sense, mirrors the axiom of intuitionistic co-reflection
\begin{equation}\label{intco}
\varphi\rightarrow K\varphi
\end{equation}
of Intuitionistic Epistemic Logic. In fact, if ``there is a construction that converts any given proof of $\varphi$ into a proof of $K\varphi$" (the BHK reading of \eqref{intco}), then that construction represents a kind of proof, and thus we may state that ``there is evidence that if $\varphi$ is proved, then $K\varphi$ is proved" (the classical reading of \eqref{classco2}). The described translation of \eqref{intco} into \eqref{classco} and \eqref{classco2} is similar to the above discussed meaning of axiom (A3) and illustrates the way we are reasoning about intuitionistic truth in our modal systems. In this sense, \eqref{classco2} can be seen as a classical representation of intuitionistic co-reflection. We refer to \eqref{classco} as the axiom of \textit{classical} co-reflection. \\

In the following, we give some examples of derivations in our modal logics. The derived theorems represent laws for the reasoning about classical knowledge and intuitionistic truth. Some of those laws correspond to related intuitionistic principles valid in $IEL$. 
However, since knowledge is classical in our modal logics, we cannot expect to find appropriate classical representations for \textit{all} intuitionistic principles of $IEL$. 

It is argued in \cite{artpro} that \textit{reflection}, $K\varphi\rightarrow\varphi$, as a law of classical knowledge, is not acceptable intuitionistically. This corresponds to the fact that in our classical modal logics, the formula $\square K\varphi\rightarrow\square\varphi$ is not derivable, as we shall see in Theorem \ref{660} below. Instead of reflection, we adopt \textit{intuitionistic reflection} $K\varphi\rightarrow\neg\neg\varphi$ as an axiom which is ``acceptable both classically and intuitionistically"  (see \cite{artpro} for a discussion). By \textit{tertium non datur}, reflection then derives as a theorem.

\begin{theorem}\label{120}
$EL3$ has classical reflection: $\vdash_{EL3}K\varphi\rightarrow\varphi$.
\end{theorem}

Theorem 7 of \cite{artpro} states that $\neg K\varphi\leftrightarrow K\neg\varphi$ is a theorem of $IEL$. Of course, the right-to-left implication is also a law of classical knowledge. The left-to-right implication, however, would represent a very strong epistemic property in a classical setting: ``for any given proposition $p$, if $p$ is not known, then its negation is known". That is, knowledge would collapse into classical truth (recall that $K\varphi\vee\neg K\varphi$ is a theorem of classical knowledge). We show that the related yet weaker statement ``if a given proposition $p$ is not known, then it is known that $p$ is not proved" holds in $EL5$.

\begin{theorem}\label{150}
$\vdash_{EL5}\neg K\varphi \rightarrow K\neg\square\varphi$.
\end{theorem}

\begin{proof}
By contrapositions of (A7) and (A2), we get $\neg K\varphi\rightarrow\neg\square\varphi$. Formula $\neg\square\varphi\rightarrow\square\neg\square\varphi$ is an instance of (A5), and $\square\neg\square\varphi\rightarrow K\neg\square\varphi$ is obtained from co-reflection combined with (A2). Transitivity of implication yields $\neg K\varphi \rightarrow K\neg\square\varphi$.\footnote{Using the fact that $\square(\varphi\rightarrow\psi)\rightarrow (\square(\psi\rightarrow\chi)\rightarrow\square(\varphi\rightarrow\chi))$ and $\square(\varphi\rightarrow\psi)\rightarrow \square(\neg\psi\rightarrow\neg\varphi)$ are theorems, one actually can show the stronger result $\vdash_{EL5}\square(\neg K\varphi \rightarrow K\neg\square\varphi)$.}
\end{proof}

In IEL, $\neg K\varphi\leftrightarrow\neg\varphi$ is valid [\cite{artpro}, Theorem 8]. The right-to-left implication is also a law of classical knowledge. By the BHK reading of the left-to-right implication, we may state that ``the impossibility of a proof of $K\varphi$ yields the impossibility of a proof of $\varphi$". Actually, this is the contraposition of intuitionistic co-reflection $\varphi\rightarrow K\varphi$ which relies on the intuition that proof yields knowledge. As argued above, intuitionistic co-reflection is, in a sense, mirrored by the $EL3^-$-theorem $\square(\square\varphi\rightarrow \square K\varphi)$. We present some further derivable modal laws:  

\begin{theorem}\label{160}
\noindent
\begin{enumerate}
\item $\vdash_{E3^-}\square\neg\square K\varphi\rightarrow \square\neg\square\varphi$. ``The impossibility of a proof of $\square K\varphi$ implies the impossibility of a proof of $\square\varphi$."
\item $\vdash_{EL4} \square\varphi\rightarrow K\square\varphi$. ``If $\varphi$ is proved, then it is known that $\varphi$ is proved." 
\item $\vdash_{EL5} \neg\square\varphi\rightarrow K\neg\square\varphi$. ``If $\varphi$ is not proved, then it is known that $\varphi$ is not proved." 
\item $\vdash_{EL5} \neg K\square\varphi\leftrightarrow \square\neg\square\varphi$. ``The fact that $\square\varphi$ is not known is equivalent to the impossibility of a proof of $\square\varphi$."
\item $\vdash_{EL5} \neg\square K\varphi\rightarrow \square\neg\square\varphi$. ``If there is no proof of $K\varphi$, then a proof of $\square\varphi$ is impossible."\footnote{Note that in $EL5$, $\neg\square K\varphi$ ``there is no proof of $K\varphi$" is actually equivalent to $\square\neg\square K\varphi$ ``a proof of $\square K\varphi$ is impossible".} 
\item $\vdash_{EL3}\square K\varphi\rightarrow\neg\square\neg\varphi$. ``If $K\varphi$ is proved, then a proof of $\varphi$ is possible."
\end{enumerate}
\end{theorem}

\begin{proof}
First, we observe that \\
\noindent (a) $\square(\varphi\rightarrow\psi)\rightarrow (\square(\psi\rightarrow\chi)\rightarrow \square(\varphi\rightarrow\chi))$ and \\
\noindent (b) $\square(\varphi\rightarrow\psi)\rightarrow\square(\neg\psi\rightarrow\neg\varphi)$ \\
are theorems (apply AN to suitable theorems of IPC and consider distribution). 
(i): By AN and (A7), $\square(\square\varphi\rightarrow \square K\varphi)$. By (b), $\square(\neg\square K\varphi\rightarrow \neg\square\varphi)$. Now apply distribution. \\
(ii): $\square\varphi\rightarrow\square\square\varphi$, $\square\square\varphi\rightarrow \square K\square\varphi$ and $\square K\square\varphi\rightarrow K\square\varphi$ are axioms of $EL4$ and yield (ii). \\
(iii): $\neg\square K\varphi\rightarrow\square\neg\square K\varphi$ is an axiom of $EL5$, and $\square\neg\square\varphi\rightarrow K\neg\square\varphi$ is a theorem of $EL3$ (by co-reflection and (A2)). Apply transitivity of implication.\\ 
(iv): By contraposition of (ii), we get $\neg K\square\varphi\rightarrow\neg\square\varphi$. $\neg\square\varphi\rightarrow\square\neg\square\varphi$ is an axiom of $EL5$. By transitivity of implication, we get the left-to-right implication. The right-to-left implication follows from (A2) along with the contraposition of the theorem $K\square\varphi\rightarrow\square\varphi$.\\
(v): This follows from the contrapostion of (A7) along with the $EL5$-axiom $\neg\square\varphi\rightarrow\square\neg\square\varphi$.\\
(vi): First, note that $\bot\rightarrow\square\bot$ and $\square\bot\rightarrow\bot$ are axioms of our modal logics: the former is a substitution-instance of the IPC theorem $\bot\rightarrow x$, and the latter is an instance of (A2). By AN, $\bot\equiv\square\bot$. SP then ensures that $\bot$ and $\square\bot$ can be replaced by each other, in every context. Since $K\varphi\rightarrow\neg\neg\varphi$ is an axiom of $EL3$, we may apply AN and obtain $\square (K\varphi\rightarrow\neg\neg\varphi)$. By distribution and MP, we derive $\square K\varphi\rightarrow\square((\varphi\rightarrow\bot)\rightarrow\bot)$. Again by distribution, together with transitivity of implication, we derive $\square K\varphi\rightarrow (\square(\varphi\rightarrow\bot)\rightarrow\bot)$, i.e., $\square K\varphi\rightarrow\neg\square\neg\varphi$.
\end{proof}

In logic IEL, ``no truth is unverifiable", i.e. $\neg(\neg K\varphi\wedge\neg K\neg\varphi)$ is a theorem [\cite{artpro}, Theorem 9]. The classical reading would yield: ``for any proposition $p$, either $p$ is known or its negation is known", a condition which again would imply the equivalence of truth and knowledge. We are able to derive the following weaker condition in logic $EL5$: 

\begin{theorem}\label{170}
$\vdash_{EL5} K\square\varphi\vee K\neg\square\varphi$.
\end{theorem}

\begin{proof}
By \textit{tertium non datur}, $\square\varphi\vee\neg\square\varphi$. By Theorem \ref{160} (ii) and (iii), $(\square\varphi\rightarrow K\square\varphi)\wedge(\neg\square\varphi\rightarrow K\neg\square\varphi)$ is a theorem of $EL5$. The formula $((\square\varphi\rightarrow K\square\varphi) \wedge (\neg\square\varphi\rightarrow K\neg\square\varphi))\rightarrow ((\square\varphi\vee\neg\square\varphi)\rightarrow  (K\square\varphi\vee  K\neg\square\varphi))$ is a substitution-instance of the intuitionistic theorem $((x_1\rightarrow y_1)\wedge (x_2\rightarrow y_2))\rightarrow ((x_1\vee x_2)\rightarrow (y_1\vee y_2))$ and is therefore an axiom of $EL5$. Now we may apply Modus Ponens two times and obtain $K\square\varphi\vee K\neg\square\varphi$.
\end{proof}  

\section{Algebraic semantics} 

We expect the reader to be familiar with some basic lattice-theoretical notions such as (prime, ultra-) filters on lattices. We adopt the notation from \cite{lewjlc1,lewjlc2} and write a bounded lattice as an algebraic structure $\mathcal{H}=(H,f_\bot,f_\top,f_\vee,f_\wedge)$, where $f_\bot,f_\top$ are the least and the greatest elements w.r.t. the induced lattice order, and $f_\vee, f_\wedge$ are the binary operations for join and meet, respectively. Recall that a Heyting algebra can be defined as a bounded lattice together with an additional binary operation for implication $f_\rightarrow$ which maps any two elements $m,m'$ to the supremum of $\{m''\mid f_\wedge(m,m'')\le m'\}$, where $\le$ is the lattice order. That supremum $f_\rightarrow(m,m')$ is also called the relative pseudo-complement of $m$ with respect to $m'$. The pseudo-complement $f_\neg(m)$ of an element $m$ then is defined by $f_\neg(m)=f_\rightarrow(m,f_\bot)$. As in \cite{lewjlc2}, we say that a Heyting algebra has the \textit{Disjunction Property} DP if its smallest filter $\{f_\bot\}$ is a prime filter. Note that in a Heyting algebra with DP, the equation $f_\vee(m,m')=f_\top$ is equivalent to the condition that $m=f_\top$ or $m'=f_\top$. A Boolean algebra $\mathcal{B}$ has DP if and only if $\mathcal{B}$ has at most two elements.\\

Our algebraic models are given by certain Heyting algebras with a designated ultrafilter as truth-set. An ultrafilter is defined as a maximal filter (which exists by Zorn's Lemma). Maximal filters reflect the \textit{classical} behavior of the logical connectives in the sense of item (b) of the next Lemma. On the other hand, intuitionistic truth can be represented by the greatest element of a Heyting algebra. In this way, we are able to model and to combine classical and intuitionistic truth within the same Heyting algebra. The following facts are crucial for our semantic modeling.

\begin{lemma}\label{200}
Let $\mathcal{H}$ be a Heyting algebra with universe $H$.\\
(a) $U\subseteq H$ is an ultrafilter iff there is a Heyting algebra homomorphism $h$ from $\mathcal{H}$ to the two-element Boolean algebra $\mathcal{B}$ such that the top element of $\mathcal{B}$ is precisely the image of $U$ under $h$.\\
(b) If $U\subseteq H$ is an ultrafilter, then for all $m,m'\in H$:
\begin{itemize}
\item $f_\vee(m,m')\in U$ iff $m\in U$ or $m'\in U$ (i.e. $U$ is a prime filter)
\item $m\in U$ or $f_\neg(m)\in U$
\item $f_\rightarrow(m,m')\in U$ iff [$m\notin U$ or $m'\in U$] iff $f_\vee(f_\neg(m),m')\in U$.
\end{itemize}
\end{lemma}

\begin{proof}
(a) is a well-known property of Heyting algebras, (b) follows straightforwardly from (a).
\end{proof}

\begin{definition}\label{205}
An $EL3^-$-model is a Heyting algebra 
\begin{equation*}
\mathcal{M}=(M, \mathit{TRUE}, \mathit{BEL}, f_\bot, f_\top, f_\vee, f_\wedge, f_\rightarrow, f_\square, f_K)
\end{equation*}
with a designated ultrafilter $\mathit{TRUE}\subseteq M$, a set $\mathit{BEL}\subseteq M$ and additional unary operations $f_\square$ and $f_K$ such that for all $m,m',m''\in M$ the following truth conditions are fulfilled (as before, $\le$ denotes the lattice order):
\begin{enumerate}
\item $f_\square(f_\vee(m,m'))\le f_\vee(f_\square(m),f_\square(m'))$
\item $f_\square(m)\le m$
\item $f_\square(f_\rightarrow(m,m'))\le f_\square(f_\rightarrow(f_\square(m),f_\square(m')))$
\item $f_\square(m)\in\mathit{TRUE}\Leftrightarrow m=f_\top$
\item $f_K(m)\in\mathit{TRUE}\Leftrightarrow m\in\mathit{BEL}$
\item $f_K(f_\rightarrow(m,m'))\le f_\rightarrow(f_K(m),f_K(m'))$
\item $f_\square(m)\le f_\square(f_K(m))$
\end{enumerate}
An $EL3$-model is an $EL3^-$-model that satisfies the additional truth condition (viii): $f_K(m)\le f_\neg(f_\neg(m))$, for all propositions $m$.

We regard $M$ as a \textit{propositional universe} and $\mathit{TRUE}\subseteq M$ as the set of propositions which are classically true. The propositions $f_\top$, $f_\bot$ represent intuitionistic truth and intuitionistic falsity, respectively. $\mathit{BEL}$ is the set of believed propositions. If $\mathit{BEL}\subseteq \mathit{TRUE}$, then we identify belief with knowledge. 
\end{definition}

Recall that in any Heyting algebra: $m\le m' \Leftrightarrow f_\rightarrow(m,m')=f_\top$. Now observe that the truth conditions (i)--(iii) and (vi)--(viii) of the above definition correspond to applications of rule AN to the axioms (A1)--(A3) and (A6)--(A8), respectively. Truth condition (iv) establishes the relation between intuitionistic and classical truth via the necessity operator. Similarly, truth condition (v) defines the relation between belief and classical truth.\\ 

Note that truth condition (viii) of an $EL3$-model, together with truth condition (v) and Lemma \ref{200}, implies that $\mathit{BEL}\subseteq\mathit{TRUE}$. That is, believed propositions are classically true and belief is knowledge in any $EL3$-model. On the other hand, in any model, the condition $\mathit{BEL}\subseteq\mathit{TRUE}$ is equivalent to the condition $f_\rightarrow(f_K(m),m)\in\mathit{TRUE}$, for all propositions $m$. From the definition of model theoretic satisfaction below it will follow that models with the property $\mathit{BEL}\subseteq\mathit{TRUE}$ are precisely the models satisfying classical reflection $K\varphi\rightarrow\varphi$.

The truth conditions (i) and (iv) ensure that every model has the Disjunction Property DP: for all $m,m'\in M$, $f_\vee(m,m')=f_\top$ iff $m=f_\top$ or $m'=f_\top$.\\ 

\begin{definition}\label{210}
Let $\mathcal{M}$ be an $EL3$-model. We say that 
\begin{itemize}
\item $\mathcal{M}$ is an $EL4$-model if for all $m\in M$: $f_\square(m)\le f_\square(f_\square(m))$\footnote{Note that this, together with (iv), implies the weaker condition (4'): $f_\square(m)= f_\top\Leftrightarrow m= f_\top$. In a discussion on the last pages of our article \cite{lewjlc1}, there is a minor incorrectness which, however, has no impact on main results and is easily corrected as follows. It is claimed there correctly that condition (4') ensures soundness of $S4$-axiom $\square\varphi\rightarrow\square\square\varphi$. However, to warrant also a sound application of AN, i.e. soundness of $\square(\square\varphi\rightarrow\square\square\varphi)$, one has to impose the stronger truth condition $f_\square(m)\le f_\square(f_\square(m))$ on algebraic $S4$-models. Unfortunately, this detail was overlooked in Corollary 5.7 of \cite{lewjlc1}.} 
\item $\mathcal{M}$ is an $EL5$-model if for all $m\in M$ the following holds:
\begin{equation*}
\begin{split}
f_\square(m)=
\begin{cases}
f_\top, \text{ if }m=f_\top\\
f_\bot, \text{ else}
\end{cases}
\end{split}
\end{equation*}
\end{itemize}
\end{definition}

Note that the defining condition of $EL5$-models implies $f_\square(m)\le f_\square(f_\square(m))$ and $f_\neg(f_\square(m))\le f_\square(f_\neg(f_\square(m)))$, for all propositions $m$. In particular, every $EL5$-model is an $EL4$-model.

\begin{definition}\label{220}
Let $\mathcal{L}\in\{EL3^-,EL3,EL4,EL5\}$. An assignment of an $\mathcal{L}$-model $\mathcal{M}$ is a function $\gamma\colon V\rightarrow M$ that extends in the canonical way to a function $\gamma\colon Fm\rightarrow M$: $\gamma(\bot)=f_\bot$, $\gamma(\top)=f_\top$, $\gamma(\square\varphi)=f_\square(\gamma(\varphi))$, $\gamma(K\varphi)=f_K(\gamma(\varphi))$, $\gamma(\varphi*\psi)=f_*(\gamma(\varphi),\gamma(\psi))$, for $*\in\{\vee,\wedge,\rightarrow\}$. An $\mathcal{L}$-interpretation is a tuple $(\mathcal{M},\gamma)$ consisting of an $\mathcal{L}$-model and an assignment. The relation of satisfaction is defined by 
\begin{equation*}
(\mathcal{M},\gamma)\vDash\varphi :\Leftrightarrow \gamma(\varphi)\in\mathit{TRUE}
\end{equation*}
and extends in the usual way to sets of formulas. The relation of logical consequence is defined by $\Phi\Vdash_\mathcal{L}\varphi :\Leftrightarrow$ $(\mathcal{M},\gamma)\vDash\Phi$ implies $(\mathcal{M},\gamma)\vDash\varphi$, for every $\mathcal{L}$-interpretation $(\mathcal{M},\gamma)$.
\end{definition}

The following is not hard to prove (see, e.g., \cite{lewjlc1}): 

\begin{lemma}\label{222}
\begin{equation*}
(\mathcal{M},\gamma)\vDash\varphi\equiv\psi\Leftrightarrow\gamma(\varphi)=\gamma(\psi).
\end{equation*}
\end{lemma}

That is, $\varphi\equiv\psi$ is true iff $\varphi$ and $\psi$ denote the same proposition -- this is precisely the intended meaning of an identity connective. By Lemma \ref{200}, the condition that the proposition $f_\leftrightarrow(m,m'):=f_\wedge (f_\rightarrow(m,m'),f_\rightarrow(m',m))$ belongs to the ultrafilter $\mathit{TRUE}$ of a given model is equivalent to the condition: $m\in\mathit{TRUE} \Leftrightarrow m'\in\mathit{TRUE}$. The latter, however, does not imply $m=m'$. Thus, the Fregean Axiom $(\varphi\leftrightarrow\psi)\rightarrow (\varphi\equiv\psi)$ is not valid. This shows that we are actually dealing with a non-Fregean semantics.\footnote{A counterexample showing that the equivalence $m\in\mathit{TRUE}\Leftrightarrow m'\in\mathit{TRUE}$ does not imply $m=m'$ is given by the model constructed in the proof of Theorem \ref{660} below. The model is based on the Heyting algebra of the closed interval of reals $[0,1]$ with (unique) ultrafilter $\mathit{TRUE}=(0,1]$. If $m\in\mathit{TRUE}$, then $f_\neg(m)=f_\rightarrow(m,0)=0$ and $f_\neg(f_\neg(m))=f_\neg(0)=1$. Furthermore, $m=0$ implies $f_\neg(f_\neg(m))=0$. Thus, $m\in\mathit{TRUE}$ iff $f_\neg(f_\neg(m))\in\mathit{TRUE}$. However, for any $m\in\mathit{TRUE}\smallsetminus\{1\}$ we have $m\neq f_\neg(f_\neg(m))=1$. The model witnesses what is intended by our non-Fregean semantics: intuitionistically equivalent formulas have the same meaning, but classically equivalent formulas, such as $\varphi$ and $\neg\neg\varphi$, may denote different propositions.}   

\section{Soundness, completeness and some consequences}

It is a well-known fact that all intuitionistic theorems evaluate to the top element in every Heyting algebra, under any assignment. Also recall that $f_\rightarrow(m,m')=f_\top \Leftrightarrow m\le m'$ is a general law in Heyting algebras. Then the soundness of axioms (INT) and (A1)--(A3) as well as the soundness of rule AN follows from properties of Heyting algebras along with the truth conditions (i)--(iii) of a model. Theorem scheme (T) of \textit{tertium non datur} is sound because $\mathit{TRUE}$ is an ultrafilter, see Lemma \ref{200} (b). Properties of ultrafilters also ensure soundness of rule MP. Furthermore, one easily checks that the additional truth conditions of an $EL3$-, $EL4$-, $EL5$-model ensure soundness of the axioms (A8), (A4), (A5), respectively -- inclusively soundness of the application of AN to those axioms. We conclude that the logics $EL3^-, EL3, EL4, EL5$ are sound with respect to the corresponding classes of models: $\Phi\vdash_\mathcal{L}\varphi$ implies $\Phi\Vdash_\mathcal{L}\varphi$, for $\mathcal{L}\in\{EL3^-, EL3, EL4, EL5\}$.\\

In the following, we consider $EL5$ as the underlying deductive system. The notions of (in)consistent and maximal consistent set of formulas are defined as usual. By standard arguments, any consistent set extends to a maximal consistent set of formulas. We show that every maximal consistent set $\Phi$ is satisfiable, i.e. there is an $EL5$-model $\mathcal{M}$ and an assignment $\gamma\colon V\rightarrow M$ such that $(\mathcal{M},\gamma)\vDash\Phi$. For a maximal consistent set $\Phi$, we define the relation $\approx_\Phi$ on $Fm$ by 
\begin{equation*}
\varphi\approx_\Phi\psi :\Leftrightarrow\Phi\vdash_{EL5}\varphi\equiv\psi.
\end{equation*} 

\begin{lemma}\label{620}
Let $\Phi\subseteq Fm$ be a maximal consistent set. The relation $\approx_\Phi$ is an equivalence relation on $Fm$ with the following properties:
\begin{itemize}
\item If $\varphi_1\approx_\Phi\psi_1$ and $\varphi_2\approx_\Phi\psi_2$, then $\neg\varphi_1\approx_\Phi\neg\psi_1$, $\square\varphi_1\approx_\Phi\square\psi_1$, $K\varphi_1\approx_\Phi K\psi_1$ and $(\varphi_1\rightarrow\varphi_2)\approx_\Phi(\psi_1\rightarrow\psi_2)$.
\item If $\varphi\approx_\Phi\psi$, then $\varphi\in\Phi\Leftrightarrow\psi\in\Phi$.
\end{itemize}
\end{lemma}

\begin{proof}
The first item follows from the Substitution Property SP (see Lemma \ref{100}), the second item follows from the fact that $\varphi\approx_\Phi\psi$ implies $\Phi\vdash_{EL5}\varphi\leftrightarrow\psi$. 
\end{proof}

\begin{lemma}\label{630}
If $\Psi\subseteq Fm$ is consistent, then there is an $EL5$-interpretation $(\mathcal{M},\gamma)$ such that $(\mathcal{M},\gamma)\vDash\Psi$.
\end{lemma}

\begin{proof}
Let $\Psi$ be consistent. By Zorn's Lemma, $\Psi$ is contained in a maximal consistent set $\Phi$. For $\varphi\in Fm$, let $\overline{\varphi}$ be the equivalence class of $\varphi$ modulo $\approx_\Phi$. We define:
\begin{itemize}
\item $M=\{\overline{\varphi}\mid\varphi\in Fm\}$ 
\item $\mathit{TRUE}=\{\overline{\varphi}\mid\varphi\in \Phi\}$
\item $\mathit{BEL}=\{\overline{\varphi}\mid K\varphi\in \Phi\}$
\item functions $f_\top, f_\bot$, $f_\square$, $f_*$, where $*\in\{\vee,\wedge,\rightarrow\}$, by $f_\top=\overline{\top}$, $f_\bot=\overline{\bot}$, $f_\square(\overline{\varphi})=\overline{\square\varphi}$, $f_K(\overline{\varphi})=\overline{K\varphi}$, $f_*(\overline{\varphi},\overline{\psi})=\overline{\varphi * \psi}$, respectively.
\end{itemize}
In the same way as in the proof of the corresponding Lemma 4.2 of \cite{lewjlc2}, one shows that $\mathcal{M}$ is a Heyting algebra with ultrafilter $\mathit{TRUE}$ such that the truth conditions of a model hold: $\mathcal{M}$ is a Heyting algebra because all intuitionistic theorems are contained in $\Phi$, in particular those of the form $\varphi\leftrightarrow\psi$. Rule AN then yields $\Phi\vdash_{EL5}\varphi\equiv\psi$. Thus, $\mathcal{M}$ satisfies the equations which axiomatize the class of Heyting algebras. By Lemma \ref{200}, $\mathit{TRUE}$ is an ultrafilter. Truth condition (iv) of Definition \ref{205} is warranted by Lemma \ref{110}. Truth condition (v) holds by definition of the set $\mathit{BEL}$. The remaining truth conditions of a model follow from the corresponding axioms, all contained in $\Phi$, along with applications of rule AN. We show that the defined model is an $EL5$-model, i.e. we verify the second condition of Definition \ref{210}. It is enough to check that $f_\square(\overline{\varphi})\in\mathit{TRUE}$ implies $f_\square(\overline{\varphi})=f_\top$, and $f_\square(\overline{\varphi})\notin\mathit{TRUE}$ implies $f_\square(\overline{\varphi})=f_\bot$. By axiom (A4) and truth condition (iv) of a model, we get the following implications: $f_\square(\overline{\varphi})\in\mathit{TRUE}$ $\Rightarrow$ $f_\square(f_\square(\overline{\varphi}))\in\mathit{TRUE}$ $\Rightarrow$ $f_\square(\overline{\varphi})=f_\top$. By axiom (A5) and truth condition (iv) of a model, we obtain the following implications: $f_\square(\overline{\varphi})\notin\mathit{TRUE}$ $\Rightarrow$ $f_\neg(f_\square(\overline{\varphi}))\in\mathit{TRUE}$ $\Rightarrow$ $f_\square (f_\neg(f_\square(\overline{\varphi})))\in\mathit{TRUE}$ $\Rightarrow$ $f_\neg(f_\square(\overline{\varphi}))=f_\top$ $\Rightarrow$ $f_\square(\overline{\varphi}))=f_\bot$. We have proved that $\mathcal{M}$ is an $EL5$-model. Finally, we define the assignment $\gamma\colon V\rightarrow M$ by $x\mapsto\overline{x}$. Then it follows by induction on the complexity of formulas that $\gamma(\varphi)=\overline{\varphi}$, for any formula $\varphi$. Hence,
\begin{equation*} 
\varphi\in\Phi\Leftrightarrow\overline{\varphi}\in \mathit{TRUE}\Leftrightarrow\gamma(\varphi)\in \mathit{TRUE}\Leftrightarrow (\mathcal{M},\gamma)\vDash\varphi.
\end{equation*}
\end{proof}

By CPC, $\Phi\nvdash_{EL5}\varphi$ implies that $\Phi\cup\{\neg\varphi\}$ is consistent in $EL5$. This yields the Completeness Theorem.

\begin{corollary}\label{650}
Let $\mathcal{L}\in\{EL3^-, EL3, EL4, EL5\}$ and $\Phi\cup\{\varphi\}\subseteq Fm$. Then $\Phi\Vdash_\mathcal{L}\varphi$ if and only if $\Phi\vdash_\mathcal{L}\varphi$.
\end{corollary}

In \cite{artpro}, the principle of intuitionistic reflection ``given a proof of $K\varphi$, one can construct a proof of $\varphi$" is rejected. In fact, it is shown [\cite{artpro}, Theorem 5] that $K\varphi\rightarrow\varphi$ is not derivable in Intuitionistic Epistemic Logic. In the following, we prove a corresponding result for our modal logics.

\begin{theorem}\label{660}
$\nvdash_{EL5}\square K\varphi\rightarrow\square\varphi$.
\end{theorem}

\begin{proof}
It is enough to construct an $EL5$-model that satisfies $\square K\varphi$ and $\neg\square\varphi$. Then $\nVdash_{EL5}\square K\varphi\rightarrow\square\varphi$ and the assertion follows from soundness of logic $EL5$. We consider the closed interval $M=[0,1]$ of real numbers from $0$ to $1$ and the semi-open subset $\mathit{TRUE}=(0,1]=M\smallsetminus\{0\}$. One easily checks that the natural order on $M$ induces a Heyting algebra with disjunction property where meet and join are the operations for infimum and supremum, respectively, and implication is given by $f_\rightarrow(m,m')=1$ if $m\le m'$, and $f_\rightarrow(m,m')=m'$ if $m>m'$. Of course, $\mathit{TRUE}$ is an ultrafilter on $M$. Let $b\in M$ be any real with $0< b<1$. Then we put $BEL=[b,1]$. We define operations $f_\square$ and $f_K$ on $M$ as follows: 
\begin{equation*}
\begin{split}
f_\square(m):=
\begin{cases}
1, \text{ if }m=1\\
0, \text{ else}
\end{cases}
\end{split}
\end{equation*}
and
\begin{equation*}
\begin{split}
f_K(m):=
\begin{cases}
1, \text{ if }m\in\mathit{BEL}\\
0, \text{ else.}
\end{cases}
\end{split}
\end{equation*}

Then it is not hard to verify that all truth conditions of an $EL5$-model are fulfilled. We show this in detail only for truth condition (vi): $f_K(f_\rightarrow(m,m'))\le f_\rightarrow( f_K(m),f_K(m'))$, for all reals $m,m' \in [0,1]$. Of course, the condition is fulfilled if $f_K(f_\rightarrow(m,m'))=0$. So we assume $f_K(f_\rightarrow(m,m'))=1$, i.e. $f_\rightarrow(m,m')\in\mathit{BEL}$. Then it is enough to show that $f_K(m)\le f_K(m')$. Again, we may assume $f_K(m)=1$, i.e. $m\in\mathit{BEL}$. If $m\le m'$, then we get $m'\in\mathit{BEL}$ and $f_K(m')=1$. Otherwise, $m' < m$. This implies $m'=f_\rightarrow(m,m')$. Since $f_\rightarrow(m,m')\in\mathit{BEL}$, we conclude $f_K(m')=1$. Thus, truth condition (vi) holds. 

By definition of the set $\mathit{BEL}$, there exists a $m\in BEL\smallsetminus\{1\}$. Then for $x\in V$ and for any assignment $\gamma\colon V\rightarrow M$ with $\gamma(x)=m$, we have $\gamma(\square Kx)=f_\square(\gamma(Kx))=f_\square(f_K(m))=1$, and $\gamma(\square x)=f_\square(m)=0$. Thus, $(\mathcal{M},\gamma)\vDash\square Kx$ and $(\mathcal{M},\gamma)\nvDash\square x$.
\end{proof}

Theorem \ref{660} can also be seen as a \textit{model existence theorem}. In fact, it was not so clear whether there exist (non-trivial) models for our epistemic logics. Note that the model construction presented in the completeness proof does not serve as a \textit{model existence theorem} since it presupposes the consistency of the underlying deductive system. Consistency, in turn, follows from the existence of models ... .

\begin{theorem}\label{665}
$\nvdash_{EL5}K\varphi\vee K\neg\varphi$.
\end{theorem}

\begin{proof}
The model constructed in the proof of Theorem \ref{660} is a counterexample: Choose any $m\in M$ with $0<m<b$. Then $m\notin\mathit{BEL}$ and $f_\neg(m)=f_\rightarrow(m,f_\bot)=f_\bot=0\notin\mathit{BEL}$. Hence, for $\gamma(x)=m$, we have $(\mathcal{M},\gamma)\nvDash Kx\vee K\neg x$.
\end{proof}

$K(\varphi\vee\psi)\rightarrow (K\varphi\vee K\psi)$ is not valid in Intuitionistic Epistemic Logic [\cite{artpro}, Theorem 10] neither it holds in our epistemic logics:

\begin{theorem}\label{670}
$\nvdash_{EL5}K(\varphi\vee\psi)\rightarrow (K\varphi\vee K\psi)$
\end{theorem}

\begin{proof}
We consider the Lindenbaum-Tarski algebra of IPC. This is a Heyting algebra with disjunction property -- in fact, the top element $f_\top$ is the class of all intuitionistic theorems, and the bottom element $f_\bot$ is the class of all intuitionistic contradictions. Let $\mathit{TRUE}$ be any ultrafilter and let $BEL\subseteq\mathit{TRUE}$ be any filter which is not prime. Of course, such filters exist: consider, e.g., the intersection of two distinct ultrafilters. We define the operations $f_\square$ and $f_K$ as in the proof of Theorem \ref{660}, where, of course, $0$ and $1$ are replaced by $f_\bot$ and $f_\top$, respectively. Then all truth conditions of an $EL5$ model are fulfilled. Again, we show this in detail only for truth condition (vi): $f_K(f_\rightarrow(m,m'))\le f_\rightarrow( f_K(m),f_K(m'))$, for all propositions $m$. We may assume that $f_K(f_\rightarrow(m,m'))=f_\top$, i.e. $f_\rightarrow(m,m')\in\mathit{BEL}$. It suffices to show that $f_K(m)\le f_K(m')$. We may assume $f_K(m)=f_\top$, i.e. $m\in \mathit{BEL}$. By definition of the relative pseudo-complement in a Heyting algebra, we have $f_\wedge(m,f_\rightarrow(m,m'))\le m'$. But $\mathit{BEL}$ is a filter containing $m$ and $f_\rightarrow(m,m')$. It follows that $m'\in\mathit{BEL}$, i.e. $f_K(m')=f_\top$.
\end{proof}

The next result is an adaption of [\cite{lewjlc2}], Theorem 5.1].

\begin{theorem}\label{720}
Let $\mathcal{L}\in\{EL3^-, EL3, EL4, EL5\}$ and let $\Phi\cup\{\chi\}\subseteq Fm_0$ be a set of propositional formulas. Then
\begin{equation*}
\square\Phi\vdash_\mathcal{L}\square\chi\Leftrightarrow\Phi\vdash_{IPC}\chi.
\end{equation*}
\end{theorem}

\begin{proof}
The right-to-left direction can be shown, e.g., by induction on the length of derivations. In order to show the left-to-right direction, we assume $\Phi\nvdash_{IPC}\chi$ and construct an $EL5$-model $(\mathcal{M},\varepsilon)$ such that $(\mathcal{M},\varepsilon)\vDash\square\Phi$ and $(\mathcal{M},\varepsilon)\nvDash\square\chi$. Then $\square\Phi\nVdash_{EL5}\square\chi$ and, by soundness, $\square\Phi\nvdash_{EL5}\square\chi$. The model is constructed exactly in the same way as in [\cite{lewjlc2}], Theorem 5.1], but with the additional ingredients of the set $\mathit{BEL}$ of known propositions and the function $f_K$. If we define these as $\mathit{BEL}=\{f_\top\}$ and $f_K(m)=f_\square(m)$, for all propositions $m$, then all truth conditions of an $EL5$-model are satisfied. 
\end{proof} 

Theorem \ref{720} ensures that the map $\varphi\mapsto\square\varphi$ is an embedding of IPC into the modal systems presented in this paper. Of course, that mapping does not extend to an embedding of $IEL$ into our modal logics. In fact, if one adds \textit{tertium non datur} to $IEL$, then intuitionistic knowledge collapses into classical truth. 

A prominent method to interpret IPC through a classical logic is G\"odel's translation of IPC into modal system $S4$ via the map $\varphi\mapsto tr(\varphi)$, where $tr(\varphi)$ is the result of prefixing each sub-formula of $\varphi$ with $\square$. G\"odel considered $S4$ as a calculus of classical provability. G\"odel's translation extends to embeddings of $IEL^-$ and $IEL$ into $S4V^-$ and $S4V$, respectively, where the latter are bi-modal logics augmenting $S4$ with epistemic axioms for a verification modality $V$. The BHK reading of intuitionistic epistemic principles then is explained through the classical logics $S4V^-$ and $S4V$ (see \cite{pro, artpro}). A comparison of those bi-modal logics with the present approach could be a promising task for future work. An interesting feature of our modal logics is that the modal operator $\square$ plays the role of a truth predicate for \textit{intuitionistic} truth. In fact, under the assumption that the proposition denoted by $\top$ stands for intuitionistic truth, the scheme $\square\varphi\leftrightarrow (\varphi\equiv\top)$ of Lemma \ref{110} can be seen as an adaption of the well-known Tarski-biconditionals, expressed here in the object language of our classical modal logics:\\ 
$\square\varphi$ (i.e. ``$\varphi$ is intuitionistically true") iff $\varphi$ holds intuitionistically.\footnote{The biconditionals $\square\varphi\leftrightarrow (\varphi\equiv\top)$ are also valid in the Lewis-style modal systems $S1+SP$,  $S3$, $S4$ and $S5$. However, formula $\top$ does not stand for intuitionistic truth in those systems --  in fact, $\square(\varphi\vee\psi)\rightarrow(\square\varphi\vee\square\psi)$ is not a theorem.}\\

The Disjunction Property of IPC is mirrored in our classical logics:

\begin{corollary}\label{800}
Suppose $\mathcal{L}\in\{L3, EL3^-, EL3, EL4, EL5\}$ and $\varphi,\psi\in Fm_0$. The following restricted disjunction property holds:
\begin{equation*}
\vdash_\mathcal{L}\square\varphi\vee\square\psi\text{ }\Rightarrow\text{ }\vdash_\mathcal{L}\square\varphi\text{ or }\vdash_\mathcal{L}\square\psi.
\end{equation*}
\end{corollary}

\begin{proof}
Suppose $\vdash_\mathcal{L}\square\varphi\vee\square\psi$. Of course, $\varphi\rightarrow(\varphi\vee\psi)$ and $\psi\rightarrow(\varphi\vee\psi)$ are theorems of IPC. By AN and distribution, $\square\varphi\rightarrow\square(\varphi\vee\psi)$ and $\square\psi\rightarrow\square(\varphi\vee\psi)$. Again by IPC, we obtain $(\square\varphi\vee\square\psi)\rightarrow\square(\varphi\vee\psi)$. Now we apply Theorem \ref{720} and the Disjunction Property of IPC.
\end{proof}

\section{Algebraic semantics for $IEL^-$ and $IEL$}

The goal of this final section is to modify our algebraic semantics towards a semantics for the logics $IEL^-$ and $IEL$ designed in \cite{artpro} where Kripke-style semantics is presented. First, we observe that dropping the scheme of \textit{tertium non datur} and adding the axiom scheme $\varphi\rightarrow\square\varphi$ to $EL3^-$ and $EL3$ would result in systems which are essentially equivalent to $IEL^-$ and $IEL$, respectively. In fact, by rule AN, $\varphi\equiv\square\varphi$ then would be a theorem and, by SP, the formulas $\varphi$ and $\square\varphi$ could be replaced by each other in every context. The semantical counterparts of those modifications are the following. First, without \textit{tertium non datur}, the set $\mathit{TRUE}$ of true propositions of a model (see Definition \ref{205}) is no longer required to be an ultrafilter, but only a prime filter. Second, the theorem $\varphi\equiv\square\varphi$ corresponds to the semantic condition: $f_\square(m)=m$, for all propositions $m$ of a given model. Truth condition (iv) of a model then forces the equality $\mathit{TRUE}=\{f_\top\}$, i.e. $\mathit{TRUE}$ is the smallest prime filter of the underlying Heyting algebra and classical truth becomes intuitionistic truth. The conditions concerning the set $BEL\subseteq M$ of believed propositions remain unchanged. Consequently, we may define an algebraic model for $IEL^-$ as follows:

\begin{definition}\label{1000}
An $IEL^-$-model is a Heyting algebra 
\begin{equation*}
\mathcal{M}=(M, \mathit{BEL}, f_\bot, f_\top, f_\vee, f_\wedge, f_\rightarrow, f_K)
\end{equation*}
with a set $\mathit{BEL}\subseteq M$ and an additional unary operation $f_K$ such that for all propositions $m,m'\in M$ the following truth conditions hold:
\begin{enumerate}
\item $f_\top\in\mathit{BEL}$
\item $f_K(m)=f_\top\Leftrightarrow m\in\mathit{BEL}$
\item $m\le f_K(m)$
\item $f_K(f_\rightarrow(m,m'))\le f_\rightarrow(f_K(m),f_K(m'))$
\item $f_\vee(m,m')=f_\top$ $\Rightarrow$ ($m=f_\top$ or $m'=f_\top$)
\end{enumerate}
If additionally $f_K(m)\le f_\neg(f_\neg(m))$ holds for all $m\in M$, then we call $\mathcal{M}$ an $IEL$-model.
\end{definition}

Note that truth condition (v) is the Disjunction Property DP which ensures that the smallest filter $\{f_\top\}$ of the underlying Heyting algebra is prime. Recall that DP is implicitly satisfied by all $EL3^-$-models. Also truth condition (i), $f_\top\in \mathit{BEL}$, is an implicit property of each $EL3^-$-model. \\

The language of Intuitionistic Epistemic Logic does not contain the symbol $\square$. In the following we shall work with the set of formulas $Fm_e=\{\varphi\in Fm\mid$ symbol $\square$ does not occur in $\varphi\}$.

\begin{definition}\label{1040}
Suppose $\mathcal{M}$ is an $IEL^-$-model, $\gamma\colon V\rightarrow M$ is an assignment and $\varphi\in Fm_e$. Satisfaction is defined as follows: 
\begin{equation*}
(\mathcal{M},\gamma)\vDash\varphi :\Leftrightarrow\gamma(\varphi)=f_\top
\end{equation*}
We say that a formula $\varphi$ is valid in logic $IEL^-$ (in logic $IEL$) if $(\mathcal{M},\gamma)\vDash\varphi$, for all $IEL^-$-models (all $IEL$-models) $\mathcal{M}$ and for all corresponding assignments $\gamma\in M^V$.
\end{definition}

In contrast to our \textit{strong} completeness theorems for the classical logics $EL3^-$, $EL3$, $EL4$, $EL5$, we prove here only \textit{weak} soundness and completeness of $IEL^-$ and $IEL$ w.r.t. the proposed algebraic semantics. 

\begin{theorem}[Soundness and Completeness]\label{1060}
Let $\varphi\in Fm_e$. Then $\varphi$ is a theorem of $IEL^-$ (of $IEL$) if and only if $\varphi$ is valid in $IEL^-$ (in $IEL$).
\end{theorem}

\begin{proof}
One easily checks that an $IEL^-$-model satisfies all axioms of $IEL^-$. Also rule MP is sound, for if in some model $\gamma(\psi\rightarrow\varphi)=f_\top$ and $\gamma(\psi)=f_\top$, then $\gamma(\psi)\le\gamma(\varphi)$ and thus $\gamma(\varphi)=f_\top$. By induction on derivations, all theorems of $IEL^-$ are valid. By definition, an $IEL$-model satisfies additionally all formulas of the form $K\varphi\rightarrow\neg\neg\varphi$. Thus, $IEL$ is sound w.r.t. the class of all $IEL$-models. In order to prove completeness, we consider the Lindenbaum-Tarski algebra of logic $IEL^-$. Recall that such an algebra is given by the equivalence classes of formulas modulo the relation $\approx$ defined by $\varphi\approx\psi \Leftrightarrow \varphi\leftrightarrow\psi$ is a theorem, along with the canonical operations on the equivalence classes. In the case of our epistemic language, this yields the following operations (where $\overline{\varphi}$ denotes the equivalence class of formula $\varphi$ modulo $\approx$): $f_K(\overline{\varphi}):=\overline{K\varphi}$ and $f_*(\overline{\varphi},\overline{\psi}):=\overline{\varphi * \psi}$, for $*\in\{\vee, \wedge, \rightarrow\}$. It remains to check that $\approx$ is a congruence relation, i.e. $\varphi_1\approx\psi_1$, $\varphi_2\approx\psi_2$ implies $K\varphi_1\approx K\psi_1$ and $\varphi_1 * \varphi_2\approx\psi_1 * \psi_2$, for $*\in\{\vee, \wedge, \rightarrow\}$. Suppose $\varphi_1\leftrightarrow\psi_1$ is a theorem of $IEL^-$. Then, by co-reflection, $K(\varphi_1 \rightarrow\psi_1)$ and $K(\psi_1 \rightarrow\varphi_1)$ are theorems, too. The distribution axiom yields $K\varphi_1\leftrightarrow K\psi_1$. The remaining cases follow from properties of IPC. It follows that $\approx$ is a congruence relation. If $\varphi\leftrightarrow\psi$ is (the substitution-instance of) a theorem of IPC, then $\overline{\varphi}=\overline{\psi}$. Thus, the resulting algebra $\mathcal{M}$ satisfies the set of equations that axiomatizes the class of Heyting algebras and is therefore itself a Heyting algebra. Its top element is the congruence class $f_\top:=\overline{\top}$, i.e. the class of all $IEL^-$-theorems. In [\cite{artpro}, Theorem 12] it is shown that both $IEL^-$ and $IEL$ have the disjunction property, which means that the smallest theory, respectively, is prime. Then the Heyting algebra $\mathcal{M}$ has the Disjunction Property DP: $f_\vee(\overline{\varphi},\overline{\psi})=f_\top$ implies $\overline{\varphi}=f_\top$ or $\overline{\psi}=f_\top$. If we put $\mathit{BEL}=\{\overline{\varphi}\mid K\varphi$ is a theorem of $IEL^-\}$, then one easily verifies that $\mathcal{M}=(\{\overline{\varphi}\mid\varphi\in Fm_e\}, \mathit{BEL}, f_\bot, f_\top, f_\vee, f_\wedge, f_\rightarrow, f_K)$ satisfies all conditions of an $IEL^-$-model established in Definition \ref{1000}. Consider the assignment $\varepsilon$ defined by $x\mapsto \overline{x}$. By induction on the complexity of formulas one shows that $\varepsilon(\varphi)=\overline{\varphi}$, for any $\varphi\in Fm_e$. Thus, $(\mathcal{M},\varepsilon)\vDash\varphi\Leftrightarrow\varepsilon(\varphi)=\overline{\varphi}=f_\top\Leftrightarrow$ $\vdash_{IEL^-}\varphi$, for any $\varphi\in Fm_e$. So if $\varphi$ is not a theorem of $IEL^-$, then $\varphi$ is not satisfied by interpretation $(\mathcal{M},\varepsilon)$ and cannot be valid. We have shown completeness of $IEL^-$ w.r.t. the class of all $IEL^-$-models. Completeness of $IEL$ w.r.t. the class of all $IEL$-models follows similarly.
\end{proof}


\begin{thebibliography}{99}

\bibitem{artpro} S. Artemov and T. Protopopescu, \textit{Intuitionistic Epistemic Logic}, The Review of Symbolic Logic 9(2), 266--298, 2016.

\bibitem{blosus} S. L. Bloom and R. Suszko, \textit{Investigation into the sentential calculus with identity}, Notre Dame Journal of Formal Logic 13(3), 289 -- 308, 1972. 

\bibitem{hugcre} G. E. Hughes and M. J. Cresswell, \textit{A new introduction to modal logic}, Routledge, 1996.

\bibitem{lewsl1} S. Lewitzka, \textit{$\in_K$: A non-Fregean logic of explicit knowledge}, Studia Logica 97(2), 233--264, 2011.

\bibitem{lewjlc1} S. Lewitzka, \textit{Algebraic semantics for a modal logic close to S1}, Journal of Logic and Computation 26(5), 1769--1783, 2016.

\bibitem{lewsl} S. Lewitzka, \textit{Denotational semantics for modal systems S3--S5 extended by axioms for Propositional quantifiers and identity}, Studia Logica 103(3), 507--544, 2015.

\bibitem{lewjlc2} S. Lewitzka, \textit{A modal logic amalgam of classical and intuitionistic propositional logic}, Journal of Logic and Computation 27(1), 201--212, 2017. 

\bibitem{pro} T. Protopopescu, \textit{Intuitionistic Epistemology and Modal Logics of Verification}, In: Proceedings. Logics, Rationality and Interaction (LORI 2015). (Oct. 26--29, 2015). Ed. by Wiebe van der Hoek and Wesley Holliday. Lecture Notes in Computer Science 9394. Tapei: Springer, 2015, pp. 295--307.

\bibitem{sus} R. Suszko, \textit{Identity Connective and Modality}, Studia Logica 27, 7--39, 1971.

\bibitem{sus1} R. Suszko, \textit{Abolition of the fregean axiom}, Lecture Notes in Mathematics, 453:169--239 (1975), in: R. Parikh (ed.), Logic Colloquium, Springer Verlag, 2006.

\end{thebibliography}
\end{document}